\documentclass[journal]{IEEEtran}
\usepackage{fancyhdr,graphicx,amsmath,amssymb, amsfonts}
\usepackage{dsfont}
\usepackage{array}
\usepackage[caption=false,font=normalsize,labelfont=rm,textfont=rm]{subfig}
\usepackage{textcomp}
\usepackage{stfloats}
\usepackage{url}
\usepackage{verbatim}
\usepackage{graphicx}
\usepackage[nocompress]{cite}
\usepackage{xcolor}
\usepackage{amsthm}
\usepackage[ruled,vlined,linesnumbered]{algorithm2e}
\usepackage{nomencl}
\usepackage{algpseudocode}
\usepackage{makecell}
\usepackage{multicol, multirow}
\usepackage{kotex}

\newtheorem{lemma}{Lemma}
\newtheorem{proposition}{Proposition}
\newtheorem{corollary}{Corollary}

\newtheorem{remark}{Remark}
\usepackage[colorlinks=true, linkcolor=blue, citecolor=blue, urlcolor=blue]{hyperref}

\begin{document}
%
\title{Low-Complexity Blind SNR Estimator for mmWave Multi-Antenna Communications}
%
%
%

\author{Hanyoung Park,~\IEEEmembership{Graduate Student Member,~IEEE,}
        Homin Jang, 
        and~Ji-Woong Choi,~\IEEEmembership{Senior Member,~IEEE}
\thanks{This work was supported by Institute of Information \& communications Technology Planning \& Evaluation (IITP) grant funded by the Korea government (MSIT) (No. RS-2024-00442085, No. RS-2024-00398157). 
Part of this work was presented at the 2025 IEEE Global Communications Conference Workshops (GC Wkshps) \cite{ref:confproceeding}. This article extends the authors’ work by (i) including a dynamic adaptation of the threshold parameter, which partitions the signal-relevant and noise-dominated elements, (ii) adding a theoretical description for noise power estimation, and (iii) conducting a hardware implementation.
\textit{(Corresponding author: Ji-Woong Choi.)}}
\thanks{H. Park is with the Department
of Electrical Engineering and Computer Science, Daegu Gyeongbuk Institute of Science and Technology (DGIST), Daegu 42988,
 South Korea (e-mail: prkhnyng@dgist.ac.kr).}
 \thanks{H. Jang and J.-W. Choi are with the Department
of Electrical Engineering and Computer Science, Daegu Gyeongbuk Institute of Science and Technology (DGIST), Daegu 42988,
 South Korea, and also with Brain Engineering Convergence Research Center (BCC), DGIST, Daegu 42988, South Korea (e-mail: jhm7874@dgist.ac.kr; jwchoi@dgist.ac.kr).}
}

\maketitle

\begin{abstract}
In this paper, we propose a low-complexity blind estimator for the average noise power, average signal power, and signal-to-noise ratio (SNR) in millimeter-wave (mmWave) massive multi-antenna uplink systems. In particular, the proposed method is designed to operate using only a single received signal sample, without relying on pilot signals, iterative optimization, or multiple observations, and without requiring prior knowledge of the transmitted signal. By exploiting the inherent sparsity of mmWave channels in the beamspace domain, the estimator identifies noise-dominant components through a sorting-based procedure combined with a finite-difference criterion. This separation is further supported by the order statistics of noise power under Gaussian assumptions, enabling statistically grounded discrimination between signal and noise elements. The average noise power is estimated from the identified noise-only components, and the signal power and SNR are subsequently obtained through simple arithmetic operations. The proposed algorithm achieves low computational complexity and is well-suited for real-time implementation. To demonstrate its practical feasibility, a hardware-efficient very large-scale integration (VLSI) architecture is developed and implemented on a AMD-Xilinx Kintex UltraScale+ KCU116 Evaluation Kit, with corresponding field-programmable gate array (FPGA) results provided. The implementation exhibits low latency and sublinear scaling of hardware resource utilization with respect to the number of antennas, and enables parameter estimation within a duration shorter than a single symbol of conventional wireless systems. Simulation results verify that the proposed estimator achieves high estimation accuracy compared to existing single-sample-based methods. These results demonstrate that the proposed approach enables efficient and real-time blind parameter estimation for mmWave massive multi-antenna systems.
\end{abstract}


%
\IEEEpeerreviewmaketitle

\section{Introduction}
In recent years, the rapid advancement of wireless technologies has led to a substantial increase in the demand for higher data rates~\cite{ref:ericssonmobreport2025q4}. 
The range of candidate devices in wireless communication systems has significantly expanded beyond conventional smartphones to include emerging applications such as virtual and augmented reality (VR/AR), smartglasses, diverse Internet-of-Things (IoT) devices, real-time industrial automation, and autonomous driving support systems~\cite{ref:ericsson, ref:6g}. 
These applications require not only high data rates but also stringent latency requirements, which fundamentally necessitate significantly wider bandwidth. 
However, the conventional sub-6 GHz frequency bands are already heavily congested and offer limited available spectrum, making it increasingly challenging to meet these requirements within existing frequency resources.
To address this limitation, millimeter-wave (mmWave) systems have been considered a key enabler for next-generation wireless communications due to their ability to provide abundant bandwidth~\cite{ref:mmwavetextbook}. 
In addition, due to the severe path loss inherent to mmWave propagation, large-scale antenna systems, such as massive multi-input multi-output (MIMO), are widely adopted to compensate for this limitation through beamforming gain~\cite{ref:pathloss,ref:mmwavemimo}.

Meanwhile, wireless communication environments are inherently dynamic, and maintaining reliable communication links requires accurate knowledge of various environmental parameters. Such information plays a critical role not only in signal processing techniques, such as channel estimation, equalization, and precoding~\cite{ref:channelest,ref:equalization, ref:precoding}, but also in resource allocation strategies, including scheduling, adaptive modulation and coding (AMC), and power control~\cite{ref:ofdmpower,ref:amc,ref:resourcealloc,ref:scheduling,ref:powercontrol}. In most existing frameworks, these parameters are estimated using dedicated training phases with pilot sequences, which enable accurate system characterization by observing the system via prearranged signals~\cite{ref:3gpp38214, ref:3gpp38211, ref:3gpp38213}.
However, mmWave systems exhibit unique characteristics that make conventional pilot-based approaches less effective~\cite{ref:blind, ref:shortpilot2}. Specifically, mmWave channels suffer from short channel coherence time and limited multipath components, resulting in strong spatial sparsity. At the same time, they are highly susceptible to blockage, which leads to rapid fluctuations in the channel conditions~\cite{ref:mmwavemobility}. Although increasing the frequency of pilot transmission can improve tracking performance, it inevitably introduces significant overhead, thereby reducing the effective data rate, as pilot sequences do not convey information and made of reference signals~\cite{ref:shortpilot}. This issue becomes particularly critical for parameters that vary on a very short timescale, such as the signal-to-noise power ratio (SNR), for which pilot-based tracking can be more inefficient~\cite{ref:ofdm5}.

As an alternative, blind estimation techniques can be considered. 
Blind estimators can estimate the desired system parameters without relying on dedicated training resources, so they do not require additional wireless resource allocation for system observation, leading to improved spectral efficiency and higher achievable rates. 
However, most existing blind estimators rely on multiple observations, iterative optimization procedures, or learning-based approaches~\cite{ref:cnn}. These methods often incur high computational complexity and latency, making them unsuitable for real-time tracking. This limitation is further exacerbated in mmWave massive MIMO systems, where the large number of antennas causes the computational complexity to scale accordingly~\cite{ref:complexity}. Therefore, there is a necessity for a low-complexity blind parameter estimator that can operate using only a single observation while enabling real-time adaptation to rapidly varying wireless environments.
\subsection{Prior Works}
Previous studies have proposed various blind noise power or SNR estimation methods for wireless communication systems. 
Several studies exploit inherent characteristics of orthogonal frequency division multiplexing (OFDM) \cite{ref:ofdm3, ref:ofdm2, ref:ofdm1, ref:ofdm4, ref:ofdm5}.
In particular, the work in \cite{ref:ofdm5} exploits the Gaussian-distributed nature of time-domain OFDM signals to enable blind SNR estimation based on statistical moments.
The work in \cite{ref:ofdm3} utilizes the redundancy introduced by the cyclic prefix (CP), and the works in \cite{ref:ofdm2} and \cite{ref:ofdm1} exploited the periodic structure of CP-OFDM signals via time-varying autocorrelation and block circulant structure by CP, respectively.
The work in \cite{ref:ofdm4} exploits the multichannel structure enabled by MIMO-OFDM systems.

Furthermore, statistical inference-based estimations have been proposed in previous works.
First, higher-order statistics have been utilized for blind SNR estimation using the statistical properties of Gaussian noise, including classic estimators such as the $M_1M_2$ estimator~\cite{ref:m1m2} and the $M_2M_4$ estimator~\cite{ref:m2m40}. 
These approaches have been further refined in various directions \cite{ref:m2m41,ref:m2m42,ref:m2m43,ref:m2m44, ref:m2m45}. 
In particular, the work in \cite{ref:m2m41} generalizes the $M_2M_4$ estimator using recurrence relations among higher-order moments, while the work in \cite{ref:m2m42} improves estimation performance through optimized combinations of moment ratios. 
Extensions to multi-antenna systems are considered in \cite{ref:m2m43}, where higher-order moments are exploited across both spatial and temporal dimensions. 
Joint estimation of the Ricean K-factor and SNR is addressed in \cite{ref:m2m44}, and the work in \cite{ref:m2m45} developed envelope-based estimators applicable to nonconstant modulus constellations.

Additionally, expectation-maximization (EM) algorithm-based approach is also utilized~\cite{ref:em1, ref:em2, ref:em3, ref:em4}. 
The work in \cite{ref:em1} developed an EM-based blind SNR estimation method by treating the transmitted symbols as latent variables. The work in \cite{ref:em2} proposed an EM-based maximum-likelihood (ML) estimator for blind per-antenna SNR estimation. Bellili \textit{et al.} proposed an ML-based approach for SNR estimation over time-varying channels using the EM algorithm~\cite{ref:em3, ref:em4}.

Various other approaches have also been proposed.
Mahmoud \textit{et al.} devised a blind SNR estimation method that obtains the estimate by converting the error vector magnitude (EVM)~\cite{ref:evm}.
Several works proposed deep learning-based blind SNR estimators, with spectrograms of multiple signal samples~\cite{ref:cnn, ref:spectrogram} and constellation diagram~\cite{ref:constdiagram}.
The work in \cite{ref:blindsnrhw} proposed a modulation-specific, hardware-friendly, non-iterative blind SNR estimator.
The works in \cite{ref:cov1} and \cite{ref:cov2} utilized covariance matrices of received signal in MIMO systems.
The work in \cite{ref:lowreschannelest} devised a single-sample-based blind noise power estimator for mmWave massive MIMO systems, but it requires iterative refinement.
However, these methods require multiple snapshots, implying extended observation periods, learning phases, iterative procedures, or computationally intensive operations such as correlation analysis, all of which demand a large number of signal samples and may limit their practical applicability in real-time or resource-constrained environments.
Meanwhile, a low-complexity, single-sample-based blind parameter estimator suitable for mmWave massive MIMO systems has been proposed in \cite{ref:blindest}, but it suffers from limited estimation accuracy when the underlying sparsity is not sufficiently pronounced.

\subsection{Contribution}
The rapid fluctuations of mmWave channels make fast and computationally efficient blind estimation algorithms essential. However, existing methods often require high complexity, substantial computational resources, learning phases, or long observation periods, making them unsuitable for real-time parameter tracking in mmWave large-scale multi-antenna systems due to latency and scalability constraints.
To address these limitations, in this paper, we propose low-complexity blind estimators for average noise power, average signal power, and SNR that do not rely on iterative optimization, learning-based techniques, or long observation windows. The proposed method exhibits beamspace sparsity, an inherent property of mmWave large-scale multi-antenna systems, to estimate the desired parameters using only a single signal sample. 
Furthermore, the proposed estimators are amenable to efficient very large-scale integration (VLSI) implementation, and their practicality is demonstrated through hardware-oriented design and VLSI-based validation.
The main contributions of this work are summarized as follows:
\begin{itemize}
    \item We propose a novel blind noise power estimator that utilizes the inherent sparsity of the mmWave channel in the beamspace domain and statistical characteristics of Gaussian noise. Additionally, we derive the average signal power estimate and the SNR estimate using the estimated noise power. 

    \item The proposed method leverages the sparsity of signals in the beamspace domain, interpreting signal components as outliers relative to noise-only elements. Specifically, the element-wise squared magnitudes are sorted, and the finite difference among these values is used to distinguish noise-only elements from signal-bearing elements. In the separation process, the relationship between Gaussian and exponential distributions, along with the order statistics of the exponential distribution, is exploited.

    \item A VLSI architecture is developed, and corresponding field-programmable gate array (FPGA) implementation results are presented. Furthermore, the proposed method is refined into a hardware-efficient structure by hardware-aware simplifications. The implementation results demonstrate that the proposed estimators achieve low-latency operation with high hardware efficiency.
\end{itemize}

\subsection{Paper Outline}
The rest of this paper is organized as follows. In Section~\ref{sec:systemmodel}, the system model is introduced. In Section~\ref{sec:proposed}, the proposed low-complexity blind estimators are devised. In Section~\ref{sec:simulation}, simulation results and their analyses are described. In Section~\ref{sec:hardware}, the VLSI design is proposed and the corresponding analyses are conducted.  In Secion~\ref{sec:conclusion}, this paper is concluded. 

\subsubsection*{Notation}
Boldface lowercase and uppercase letters denote vectors and matrices, respectively. For an arbitrary vector $\mathbf{a}$, its $i$-th element is denoted by $a_i$. $(\cdot)^T$ and $(\cdot)^H$ mean transpose and Hermitian transpose, respectively. $\mathbf{I}_N$ is $N\times N$ identity matrix. $|\mathbf{a}|$ and $|\mathbf{a}|^2$ present element-wise absolute-valued vector and element-wise squared absolute-valued vector for $\mathbf{a}$, respectively. $\|\cdot\|$ means the $\ell_2$ norm. $\|\cdot\|_0$ denotes the $\ell_0$ norm, which is defined as the number of nonzero elements.
$\mathrm{Exp}(\lambda)$ denotes exponential distribution with rate parameter $\lambda$, with the probability density function (PDF)
\begin{equation}
    f_\mathrm{Exp}(x, \lambda) = \begin{cases}
        \lambda e^{-\lambda x} & x\geq 0, \\
        0 & x<0.
    \end{cases}
\end{equation}
Calligraphic letters denote sets, and the cardinality of an arbitrary set $\mathcal{A}$ is denoted by $|\mathcal{A}|$.

\section{System Model}\label{sec:systemmodel}
In our model, we consider a single-cell mmWave single-input multiple-output (SIMO) uplink system. We assume that the base station (BS) is equipped with $M$-element uniform linear array (ULA).
Based on this assumption, the received signal vector $\mathbf{y}\in\mathbb{C}^M$ is presented as
\begin{equation}\label{eqn:systemmodel}
    \mathbf{y}=\mathbf{h}s+\mathbf{n}=\mathbf{x}+\mathbf{n},
\end{equation}
where $\mathbf{h}\in\mathbb{C}^M$ is the uplink channel vector, $s\in\mathbb{C}$ is the transmitted symbol value, $\mathbf{x}\in\mathbb{C}^M$ is the noiseless transmitted signal vector, and $\mathbf{n}\sim\mathcal{CN}(0, N_0 \mathbf{I}_M)$ is the additive white Gaussian noise (AWGN) vector.
Here, the channel vector is determined with the incident angle as
\begin{equation}
    \mathbf{h}=\sum _{\ell=1}^L g_{\ell} \mathbf{a}(\phi_\ell),
\end{equation}
where $L$ is the number of multipaths including a potential line-of-sight (LoS) path, $g_\ell$ is the complex channel gain, $\mathbf{a}(\cdot)\in\mathbb{C}^M$ is the steering vector, and $\phi_\ell$ is the spatial frequency which corresponds to the incident angle of $\ell$-th dominant propagation path.
Note that $L$ is small because the number of multipaths is limited due to the inherent characteristics of the mmWave propagation environment, which has limited scattering~\cite{ref:mmwavetextbook}.
The steering vector for an arbitrary spatial frequency $\phi$ is presented as
\begin{equation}
    \mathbf{a}(\phi) = [1, e^{-j\pi\phi}, e^{-j2\pi\phi} \cdots, e^{-j(M-1)\pi\phi}]^T.
\end{equation}
To utilize the spatial sparsity of the mmWave channel, it is transformed into beamspace domain. To this end, spatial discrete Fourier transform (DFT) is applied to the antenna-domain signal, and the received signal in the beamspace domain is presented as
\begin{equation}
    \bar{\mathbf{y}} = \mathbf{Fy} = \mathbf{Fh}s+\mathbf{Fn} = \mathbf{Fx}+\mathbf{Fn}
    = \bar{\mathbf{h}}s + \bar{\mathbf{n}} = \bar{\mathbf{x}}+\bar{\mathbf{n}},
\end{equation}
where $\mathbf{F}\in\mathbb{C}^{M\times M}$ is the unitary DFT matrix, $\bar{\mathbf{h}}$ is the beamspace channel vector, and $\bar{\mathbf{x}}$ and $\bar{\mathbf{n}}$ are the noiseless signal vector and AWGN vector in the beamspace domain, respectively. 
Note that $\bar{\mathbf{n}}$ is still AWGN because unitary DFT is unitary linear transformation. 
Also, $\|\mathbf{h}\|=\|\bar{\mathbf{h}}\|$ and $\|\mathbf{x}\|=\|\bar{\mathbf{x}}\|$ because $\mathbf{F}$ is unitary DFT.
By the beamspace sparsity of mmWave channel, $\bar{\mathbf{h}}$ is sparse, and it implies the sparsity of $\bar{\mathbf{x}}$~\cite{ref:mmwavetextbook}.

\section{Blind Estimators}\label{sec:proposed}
In this section, we first present the mathematical formulation of the system parameters of interest, i.e., the average noise power, average signal power, and SNR. 
Subsequently, we propose low-complexity blind estimators for these system parameters.
\subsection{Parameter Definitions}
In the considering system model, the average noise power is defined as
\begin{equation}\label{eqn:noisepowerdef}
    N_0 = \frac{1}{M} \mathbb{E}[\|\mathbf{n}\|^2].
\end{equation}
Also, the average signal power is defined as
\begin{equation}\label{eqn:signalpowerdef}
    P_x = \frac{1}{M} \mathbb{E}[\|\mathbf{x}\|^2].
\end{equation}
Based on the definitions of average noise power and signal power in \eqref{eqn:noisepowerdef} and \eqref{eqn:signalpowerdef}, the SNR $\rho$ is given by
\begin{equation}
    \rho = \frac{P_x}{N_0} = \frac{\mathbb{E}[\|\mathbf{x}\|^2]}{\mathbb{E}[\|\mathbf{n}\|^2]}.
\end{equation}

\subsection{Average Noise Power Estimator}
Here, we propose a low-complexity blind average noise power estimator, which is derived in Algorithm~\ref{alg:noisepwrest}.
The proposed method exploits a key characteristic of mmWave channels, namely spatial sparsity in the beamspace domain~\cite{ref:mmwavetextbook}. In particular, only a small subset of beamspace components contains significant signal power, while the majority exhibit negligible magnitudes, as illustrated in Fig.~\ref{fig:beamspace_channel}.
Based on this observation, the noiseless beamspace signal vector $\bar{\mathbf{x}}$ is predominantly composed of near-zero elements. Consequently, its noisy observation $\bar{\mathbf{y}}$ is expected to have a power distribution that is largely dominated by noise, except for a few components with strong signal contributions.
To leverage this property, we first sort the squared magnitudes of the noisy beamspace observations in ascending order. Here, we denote $|\bar{\mathbf{y}}|^2_\mathrm{sorted}\in\mathbb{R}^M$ as the sorted power vector. Subsequently, we classify a small number of the largest elements as signal-bearing components, while treating the remaining elements as noise-only. Finally, the average noise power is estimated by computing the mean power of these noise-only components.

\begin{figure}
    \centering
    \includegraphics[width=\linewidth]{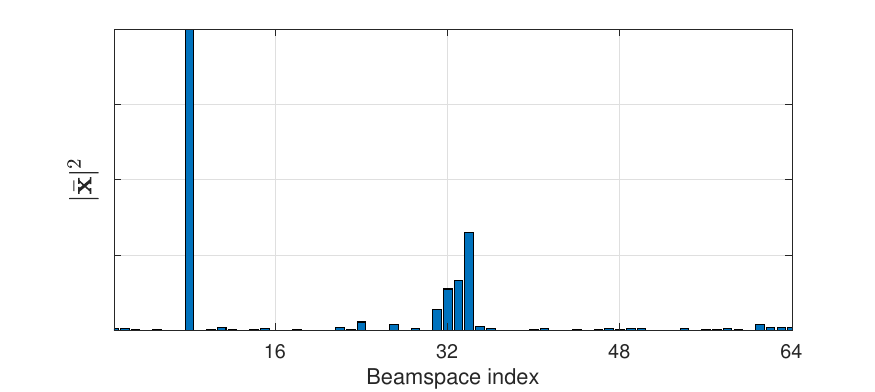}
    \caption{Squared magnitudes of the beamspace signal vector. The channel is generated with QuaDRiGa mmMAGIC model~\cite{ref:quadriga} for a carrier frequency of 50 GHz and a 64-element ULA with half-wavelength antenna spacing.}
    \label{fig:beamspace_channel}
\end{figure}

To estimate the average noise power using the proposed approach, it is necessary to establish a criterion for distinguishing between signal and noise components. This criterion is motivated by the observation that signal-dominant samples do not follow the same complex Gaussian distribution as noise and can therefore be regarded as statistical outliers. To this end, we exploit the statistical properties of Gaussian noise. For hardware-efficient implementation, we adopt a finite difference operation, which does not require complex operations such as multiplication or division.
Specifically, the finite difference is defined as
\begin{equation}
    \Delta [|\bar{\mathbf{y}}|^2_\mathrm{sorted}]_m = [|\bar{\mathbf{y}}|^2_\mathrm{sorted}]_{m+1} - [|\bar{\mathbf{y}}|^2_\mathrm{sorted}]_m,    
\end{equation}
where $[|\bar{\mathbf{y}}|^2_\mathrm{sorted}]_m$ is the $m$-th element of the sorted vector. 
This allows us to explicitly characterize the distribution of the noise component.
Since we assume that the noise is AWGN, the element-wise squared magnitudes of $\bar{\mathbf{n}}$ follow an i.i.d. exponential distribution, i.e., $|\bar{n}_m|^2\overset{\mathrm{i.i.d.}}{\sim}\mathrm{Exp}(1/N_0)$.
Leveraging this property, we further consider the order statistics, which show a well-known property of the exponential distribution.
Specifically, let $|\bar{\mathbf{n}}|^2_\mathrm{sorted}$ denote the vector of squared magnitudes of $\bar{\mathbf{n}}$, sorted in ascending order.
Then, it satisfies
\begin{equation}
    [|\bar{\mathbf{n}}|^2_\mathrm{sorted}]_1 \leq [|\bar{\mathbf{n}}|^2_\mathrm{sorted}]_2 \leq \cdots \leq [|\bar{\mathbf{n}}|^2_\mathrm{sorted}]_M,
\end{equation}
where $[|\bar{\mathbf{n}}|^2_\mathrm{sorted}]_m$ is the $m$-th smallest element.
We define its $m$-th finite difference as
\begin{equation}
    \Delta [|\bar{\mathbf{n}}|^2_\mathrm{sorted}]_m = [|\bar{\mathbf{n}}|^2_\mathrm{sorted}]_{m+1} - [|\bar{\mathbf{n}}|^2_\mathrm{sorted}]_m.    
\end{equation}
Under the assumption that each element of $|\bar{\mathbf{n}}|^2$ follows an exponential distribution, the statistical behavior of these finite differences can be characterized using the order statistics of exponential random variables. In particular, the finite differences are independent and follow exponential distributions with decreasing rates, as formally stated in the Lemma~\ref{lem:orderstat}.
\begin{lemma}\label{lem:orderstat}
Let $X_1, \dots, X_M \overset{\mathrm{i.i.d.}}{\sim} \mathrm{Exp}(\lambda)$, and denote the ordered samples by
$X_{(1)} \le \cdots \le X_{(M)}$.
Define the finite differences as
\begin{equation}
    Y_m = X_{(m+1)} - X_{(m)}, \quad m=1,\dots,M-1.
\end{equation}
Then, $\{Y_m\}$ are mutually independent and
\begin{equation}
    Y_m \sim \mathrm{Exp}((M-m)\lambda).
\end{equation}
\end{lemma}
\begin{proof}
    Please refer to \cite{ref:orderstat}.
\end{proof}
\noindent Accordingly, by substituting $1/N_0$ into $\lambda$, the finite difference can be characterized as 
\begin{equation}\label{eqn:differencedistribution}
    \Delta [|\bar{\mathbf{n}}|^2_\mathrm{sorted}]_m \sim \mathrm{Exp}\left(\frac{M-m}{N_0}\right),    
\end{equation}
and the corresponding mean and variance are given by $N_0/(M-m)$ and $N_0^2/(M-m)^2$, respectively.
Based on these properties, an intuitive approach is to identify outliers, that is, signal-containing elements, by measuring the deviation from the expected mean. However, this metric cannot be directly utilized since it requires prior knowledge of $N_0$, which is the goal of this estimation. Accordingly, in this work, we adopt a heuristic surrogate to replace this criterion without requiring explicit knowledge of $N_0$. 
To this end, we compare the finite difference with the product of a threshold parameter $\gamma$ and a temporary average, defined as the average power of the elements examined so far, which can be expressed as
\begin{equation}
    \hat{\mu}_m = \frac{1}{m}\sum_{m^\prime=1}^m [|\bar{\mathbf{y}}|^2_\textnormal{sorted}]_{m^\prime}.
\end{equation}
Assuming that only noise is present, i.e., $\hat{\mu}_m = \frac{1}{m}\sum_{i=1}^m [|\bar{\mathbf{n}}|^2_\textnormal{sorted}]_{m}$, the following relationship holds due to the properties of order statistics of the exponential distribution~\cite{ref:orderstat}:
\begin{equation}
    \mathbb{E}\left[\sum_{m^\prime=1}^m [|\bar{\mathbf{n}}|^2_\textnormal{sorted}]_{m^\prime}\right]
    = N_0 \sum_{k=1}^m \sum_{j=M-k+1}^M j^{-1}.
\end{equation}
Using the noise-only sample mean $\hat{\mu}_m$, this can be rewritten as
\begin{equation}
    N_0 = \frac{\hat{\mu}_m}{\sum_{k=1}^m \sum_{j=M-k+1}^M j^{-1}}.
\end{equation}
Substituting this into \eqref{eqn:differencedistribution}, we obtain
\begin{equation}\label{eqn:deltadistribution}
    \Delta [|\bar{\mathbf{n}}|^2_\mathrm{sorted}]_m \sim
    \mathrm{Exp}\left(\frac{(M-m)\sum_{k=1}^m \sum_{j=M-k+1}^M j^{-1}}{\hat{\mu}_m}\right).
\end{equation}
Based on this, the probability that the finite difference exceeds a given value can be evaluated, and events with sufficiently small probability can be regarded as rare occurrences. For an arbitrary observation $\delta$, the probability that this value is exceeded can be expressed as
\begin{equation}\label{eqn:probexceed}
\begin{aligned}
    &\mathrm{Pr}(\Delta [|\bar{\mathbf{n}}|^2_\mathrm{sorted}]_m \geq \delta) \\&
    = \exp\left(-\frac{(M-m)\sum_{k=1}^m \sum_{j=M-k+1}^M j^{-1}}{\hat{\mu}_m} \delta\right).
\end{aligned}
\end{equation}
If a rare occurrence is defined as an event with probability less than or equal to $\alpha$, the threshold at the $m$-th index is given by
\begin{equation}\label{eqn:thresholdderiv}
    \delta_m = -\frac{\hat{\mu}_m  \ln \alpha }{(M-m)\sum_{k=1}^m \sum_{j=M-k+1}^M j^{-1}}.
\end{equation}
Here, precise computation of this value at every step improves the reflection of the expected mean of the finite difference, but such a multiplication operation is computationally burdensome in hardware implementation.
On the other hand, using a fixed value fails to capture the variation of the expected finite difference as $m$ increases. Therefore, we adopt a piecewise constant design of $\gamma$ with three levels depending on the search index $m$, as
\begin{equation}
    \gamma = \begin{cases}
        \gamma_1,\quad\textnormal{if }m\in[1, M_1], \\
        \gamma_2,\quad\textnormal{if }m\in(M_1, M_2], \\
        \gamma_3,\quad\textnormal{otherwise.}
    \end{cases}
\end{equation}
Each value is represented by the median of \eqref{eqn:thresholdderiv} within the corresponding interval, and is quantized to powers of two to facilitate hardware implementation.
Based on the proposed thresholding strategy, a detection rule is defined to identify the transition point between noise-dominant and signal-dominant regions. Specifically, a \textit{hit} is declared at the first index $m$ that satisfies
\begin{equation}\label{eqn:origthresholdcheck}
    \Delta [|\bar{\mathbf{y}}|^2_\mathrm{sorted}]_m \geq \gamma \times \hat{\mu}_m.
\end{equation}
The corresponding index is denoted as $m^*$, which represents the boundary separating noise and signal components.
That is, with the ascending ordering, we may assume the terms meeting \eqref{eqn:origthresholdcheck} includes signal component while the previous terms counts in only the noise terms.
However, computing $\hat{\mu}_m$ at each iterative step $m$ may impose a significant computational burden. To mitigate this issue, the temporary mean $\hat{\mu}$ can be updated recursively using the value from the previous step, which can be expressed as
\begin{equation}
    \hat{\mu}_m = (\hat{\mu}_{m-1} \times (m-1) + [|\bar{\mathbf{y}}|^2_\mathrm{sorted}]_m) \times \frac{1}{m}.
\end{equation}
Still, this approach incurs a significant computational burden in terms of hardware implementation, as it requires two multiplication operations.
To address this issue, we reformulate the hit detection criterion using a surrogate inequality based on the cumulative sum, which is already required for signal power estimation, as
\begin{equation}\label{eqn:thresholdcheck}
    m\times \Delta [|\bar{\mathbf{y}}|^2_\mathrm{sorted}]_m \geq \gamma \times S_m,
\end{equation}
where $S_m$ is the temporary cumulative sum up to index $m$, given by
\begin{equation}
    S_m = m \hat{\mu}_m,
\end{equation}
and $S_M=\|\mathbf{y}\|^2$. With this approach, the need for two multiplications is reduced to one, as the temporary cumulative sum can be updated as
\begin{equation}
    S_m = S_{m-1} + [|\bar{\mathbf{y}}|^2_\mathrm{sorted}]_m,
\end{equation}
which reduces two multiplications and requires only addition, and the surrogate hit detection requires one new multiplication for $m\times \Delta [|\bar{\mathbf{y}}|^2_\mathrm{sorted}]_m$.

Using this boundary, the elements $\{1,...,m^*\}$ are regarded as noise-only elements, and the noise power is estimated as
\begin{equation}
    \widehat{N}_0 = \frac{S_{m^*}}{m^*} = \hat{\mu}_{m^*} = \frac{1}{m^*} \sum_{m=1}^{m^*} [|\bar{\mathbf{y}}|^2_\mathrm{sorted}]_m.
\end{equation}
The remaining elements are treated as signal-containing components and can be further utilized for subsequent processing, such as signal power or SNR estimation.
Finally, in the absence of any detected \textit{hit}, which may occur in low-SNR or signal-absent scenarios, a fallback strategy is employed by setting $m^* = M$, resulting in the full average being used as the noise power estimate.

\begin{algorithm}[t]
\caption{Blind Average Noise Power Estimator}\label{alg:noisepwrest}
\SetAlgoLined
\textbf{input:} $|\bar{\mathbf{y}}|^2_\mathrm{sorted}$

\textbf{initialize} $\texttt{hit} = \texttt{false}, S_0=0$.

\For{$m=1,...,M$}{ 
    Calculate the temporary sum of power
    \begin{equation}
        S_m = S_{m-1} + [|\bar{\mathbf{y}}|^2_\mathrm{sorted}]_m.\nonumber
    \end{equation}
    \eIf{$m<M$}{
        Calculate the finite difference of index $m$ 
        \begin{equation}
            \Delta [|\bar{\mathbf{y}}|^2_\mathrm{sorted}]_m = [|\bar{\mathbf{y}}|^2_\mathrm{sorted}]_{m+1} - [|\bar{\mathbf{y}}|^2_\mathrm{sorted}]_m.\nonumber
        \end{equation}

        \If{$m\times \Delta [|\bar{\mathbf{y}}|^2_\mathrm{sorted}]_m \geq \gamma\times S_m$}{
            $\texttt{hit} = \texttt{true}.$
            
            $m^*=m.$

        }
    } {
        \If{$\textnormal{\texttt{hit} = \texttt{false}}$} {
        $m^*=M$.
        }
    }
}

$\hat{\mu}_{m^*} = \frac{S_{m^*}}{m^*}.$

{$\widehat{N}_0=\hat{\mu}_{m^*}$.}

\Return {Noise power estimate $\widehat{N}_0$, total power $S_M$.}    

\end{algorithm}

The computational complexity of the proposed average noise power estimator is analyzed as follows. The algorithm consists of \textit{(i) DFT}, \textit{(ii) sorting}, \textit{(iii) finite difference computation}, \textit{(iv) search for $m^*$}, and \textit{(v) calculation of the mean power of the elements distinguished as noise}. 
A naive implementation would require repeated computation of the temporary mean at each iteration, leading to excessive complexity. 
To address this issue, steps \textit{(iii)–(v)} are integrated into a single pass, resulting in a linear complexity $\mathcal{O}(M)$. The DFT in \textit{(i)} can be efficiently implemented using the fast Fourier transform (FFT) with complexity $\mathcal{O}(M\log M)$, while the sorting step in (ii) can be performed with complexity $\mathcal{O}(M \log M)$ if it is conducted with Quicksort~\cite{ref:quicksort}. Consequently, the overall computational complexity of the proposed estimator in big-O notation is $\mathcal{O}(M\log M)$. Notably, the proposed method estimates the average power of the elements that are classified as the noise elements in the beamspace domain, without requiring prior knowledge of the ground-truth signal or pilot sequences.

\subsection{Average Signal Power Estimator and SNR Estimator}
Since the noise $\mathbf{n}$ is assumed to be AWGN, it is uncorrelated with the signal $\mathbf{x}$. Accordingly, the following relation holds:
\begin{equation}
    \|\mathbf{y}\|^2 = \|\mathbf{x}\|^2 + \|\mathbf{n}\|^2.
\end{equation}
Therefore, average signal power can be readily obtained given the average noise power. Specifically, the estimate of the average signal power can be derived from the estimate of the average noise power as
\begin{equation}\label{eqn:signalpowerest}
    \widehat{P}_x = \max\left\{ \frac{1}{M}\|\bar{\mathbf{y}}\|^2 - \widehat{N}_0, 0  \right\}.
\end{equation}
Furthermore, since SNR is defined as the ratio of the average signal power to the average noise power, the SNR estimate can be obtained as
\begin{equation}\label{eqn:snrest}
    \hat{\rho} = \frac{\widehat{P}_x}{\widehat{N}_0}.
\end{equation}

\begin{remark}
    The proposed estimators do not require knowledge of $\mathbf{s}$. Accordingly, these estimators can be operated in blind scenarios. Additionally, they do not require a training phase for machine learning or iterative optimization. Furthermore, they do not require multiple snapshots of the received signal, which reduces the complexity and feasibility of real-time parameter tracking.
\end{remark}
\begin{remark}
    The proposed estimators for the average signal power require the estimated noise power $\widehat{N}_0$. It means that these estimators can be considered as the extended applications of the noise power estimator. Moreover, this suggests that their estimation performance is influenced by that of the noise power estimator.
\end{remark}
\begin{remark}
    The system model described in \eqref{eqn:systemmodel} assumes a narrowband channel, which implies a flat-fading channel, but the proposed estimators can also be applied for wideband channels with frequency selectivity through subcarrier-wise estimation.
\end{remark}

\subsection{Theoretical Analysis}
To analyze the properties of the proposed noise power estimator, we conduct a theoretical analysis based on scenarios with an ideal sparse vector. This analysis also serves as a foundation for interpreting the simulation results presented later. First, we investigate the behavior of the estimator when an incorrect separation index is employed, i.e., when the estimator utilizes an index different from the optimal one.

\begin{lemma}\label{lem:monotone}
Consider $|\bar{\mathbf{y}}|^2_{\mathrm{sorted}}$ and $\hat{\mu}_m$ for ideally sparse vector $\bar{\mathbf{x}}_\mathrm{ideal}$ which satisfies $\|\bar{\mathbf{x}}_\mathrm{ideal}\|_0\ll M$ and an arbitrary index $m \in \{1,\ldots,M\}$.
Let $m_0$ denote the true boundary index between the noise-only and signal-bearing components, which is not available in practice. The oracle noise power estimator, defined as the estimator that uses this true boundary, is given by
\begin{equation}
    \hat{N}_0^{\mathrm{ora}} := \hat{\mu}_{m_0}
    =
    \frac{1}{m_0}\sum_{m=1}^{m_0}\left[|\bar{\mathbf{y}}|^2_{\mathrm{sorted}}\right]_m.
\end{equation}
Then, the following holds:
\begin{equation}
    \hat{N}_0 \leq \hat{N}_0^{\mathrm{ora}}, \quad \text{if } m^* < m_0,
\end{equation}
\begin{equation}
    \hat{N}_0 \geq \hat{N}_0^{\mathrm{ora}}, \quad \text{if } m^* > m_0.
\end{equation}
\end{lemma}

\begin{proof}
Please refer to Appendix \ref{app:prooflemma}.
\end{proof}

This lemma theoretically characterizes the intuitive underestimation and overestimation phenomena that arise when the separation index is chosen before or after the true index, respectively. These observations will later be used to explain the tendencies observed in the simulation results in Sec.~\ref{sec:simulation}.

Meanwhile, under the ideal condition and by leveraging the i.i.d. property of AWGN, the proposed estimator satisfies the following property:

\begin{proposition}\label{prop:unbiased}
Under the ideal sparse vector assumption in Lemma \ref{lem:monotone}, let
\begin{equation}
    k = \|\bar{\mathbf{x}}_{\mathrm{ideal}}\|_0,
\end{equation}
and let $\mathcal I_n$ and $\mathcal I_s$ denote the index sets of noise-only and signal-bearing components, respectively, where $|\mathcal I_n|=M-k$.

Assume perfect separation, i.e.,
\begin{equation}
    \max_{i\in\mathcal I_n} |\bar y_i|^2
    <
    \min_{j\in\mathcal I_s} |\bar y_j|^2.
\end{equation}
Then, the oracle noise power estimator
\begin{equation}
    \hat N_0^{\mathrm{ora}}
    =
    \frac{1}{M-k}\sum_{m=1}^{M-k}[|\bar{\mathbf y}|^2_{\mathrm{sorted}}]_m
\end{equation}
is unbiased under the AWGN assumption, i.e.,
\begin{equation}
    \mathbb E[\hat N_0^{\mathrm{ora}}]=N_0.
\end{equation}
\end{proposition}

\begin{proof}
Please refer to Appendix~\ref{app:proofprop}.
\end{proof}

This proposition shows that, under the assumption of perfect separation between signal and noise components, the proposed estimator is unbiased. Furthermore, based on this result, the variance of the estimator can be derived, which provides further insights into its statistical characteristics.

\begin{corollary}\label{cor:sparsity}
Under the same assumptions as in Proposition \ref{prop:unbiased}, the variance of the oracle estimator is given by
\begin{equation}
    \mathrm{Var}(\hat N_0^{\mathrm{ora}})
    =
    \frac{N_0^2}{M-k}.
\end{equation}
Hence, as the sparsity increases (i.e., $k$ decreases), the variance decreases. Furthermore, as $M \to \infty$ with $k \ll M$, we have
\begin{equation}
    \mathrm{Var}(\hat N_0^{\mathrm{ora}}) \to 0,
\end{equation}
which implies that the oracle estimator becomes asymptotically consistent.
\end{corollary}
\begin{proof}
    Please refer to Appendix~\ref{app:proofcor}.
\end{proof}

This indicates that estimation performance improves as the signal vector has more sparsity or as the antenna dimension $M$ increases. As these conditions are more pronounced in mmWave multi-antenna environments, the proposed approach is particularly well-suited for such settings.

\section{Simulation Results}\label{sec:simulation}
\subsection{Simulation Configuration}
For the uplink mmWave channel model, we utilize QuaDRiGa mmMAGIC UMi\cite{ref:quadriga}, which is known as a realistic mmWave channel model, with a carrier frequency of 50 GHz. We assume that the BS is equipped with a 64-element ULA. The antenna spacing of the BS is assumed to be half-wavelength, and the antenna of each user is omnidirectional. 
For the simulations, we conducted 10000 Monte-Carlo runs with regenerated channels for each trial.
Also, without loss of generality, the ground-truth average noise power is assumed to be $N_0=1$ for all simulations.

To assess the competence of the proposed algorithm, we consider baselines from recent literature, which can operate with a single signal snapshot.
For the baselines, a median absolute deviation (MAD)-based nonparametric estimator and its parametric refined variant\cite{ref:blindest}, as well as a truncated-mean-based estimator which trims outliers with the confidence interval\cite{ref:lowreschannelest}, are considered. 
Additionally, to evaluate the effectiveness of the dynamic thresholding parameter $\gamma$, we consider two variants of the proposed method: `Proposed (fixed)', which employs a fixed $\gamma$, and `Proposed (dynamic)', which uses $\gamma_1,\gamma_2$, and $\gamma_3$ depending on the sorted indices.

\subsection{Numerical Results and Analysis}
Fig.~\ref{fig:noise} illustrates the estimated average noise power as a function of SNR. 
In low-SNR regimes, the proposed algorithm slightly underestimates the noise power, whereas in high-SNR regimes, it tends to slightly overestimate it. As discussed in Lemma~\ref{lem:monotone}, this behavior arises because, under low-SNR conditions, noise components with large magnitudes are more likely to be misclassified as signal. Conversely, under high-SNR conditions, beam components with small magnitudes are more likely to be misclassified as noise.
Furthermore, the proposed method achieves higher estimation accuracy than the baseline estimator, as it performs noise power estimation after explicitly separating signal and noise components.
This improvement arises because MAD-based estimators~\cite{ref:blindest} do not perform explicit sample-wise outlier detection, leading to a positive bias when the signal sparsity is not sufficiently strong. Likewise, the ``Truncated mean''~\cite{ref:lowreschannelest} relies on confidence-interval-based trimming rather than direct sample inspection, which may also introduce bias in the estimation.
In comparison with fixed-threshold and dynamic-threshold approaches, the use of dynamic parameter selection provides a significant overall performance improvement. This result suggests that employing index-dependent thresholds enhances estimation accuracy.

\begin{figure}
    \centering
    \includegraphics[width=\linewidth]{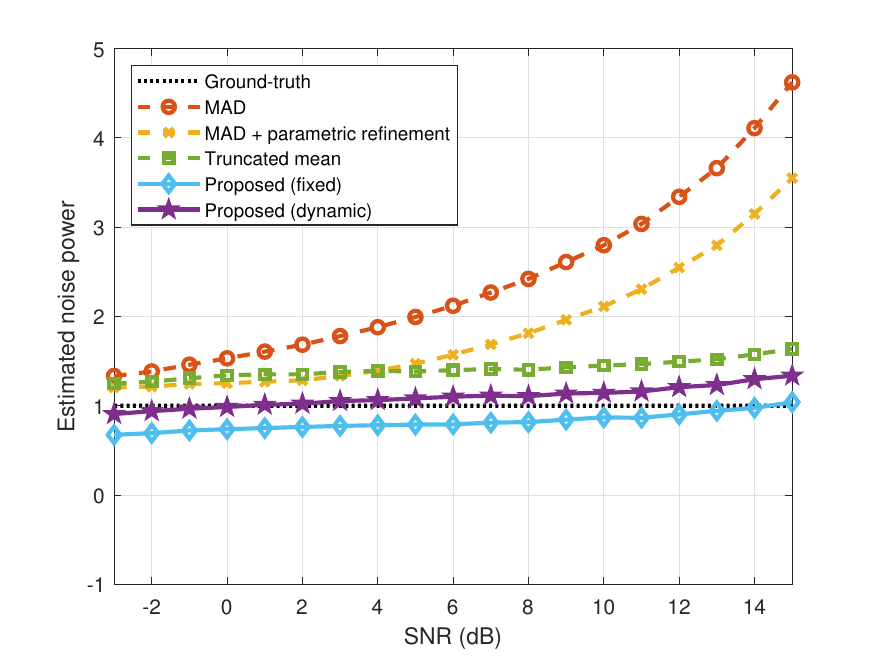}
    \caption{Estimated noise power depending on SNR.}
    \label{fig:noise}
\end{figure}

Fig.~\ref{fig:signal} shows the estimated average signal power as a function of SNR. The proposed estimators exhibit near-accurate performance across the entire SNR range, whereas the baseline methods show a slight negative offset. This discrepancy arises because the signal power estimation is influenced by the accuracy of the noise power estimator.
Fig.~\ref{fig:snr} presents the estimated SNR as a function of the ground-truth SNR. The results demonstrate that the proposed estimator achieves the highest accuracy among the considered methods. This is also because the SNR estimate depends on both the average noise power and average signal power estimates, both of which are more accurately captured by the proposed approach.
In addition, the use of a dynamic threshold parameter yields better SNR estimation performance than a fixed threshold parameter, further highlighting the advantage of index-dependent thresholding.

\begin{figure}
    \centering
    \includegraphics[width=\linewidth]{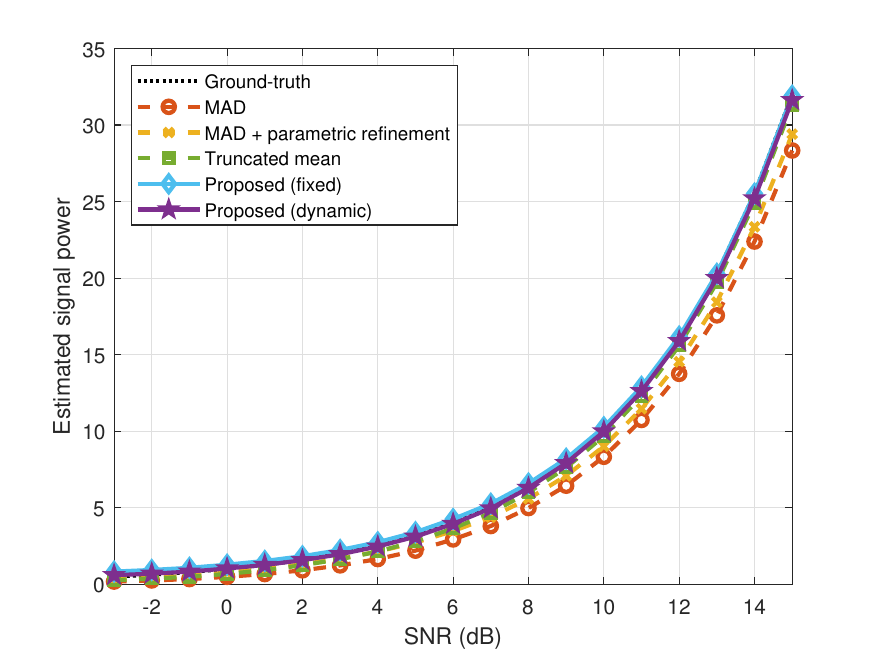}
    \caption{Estimated signal power depending on SNR.}
    \label{fig:signal}
\end{figure}

\begin{figure}
    \centering
    \includegraphics[width=\linewidth]{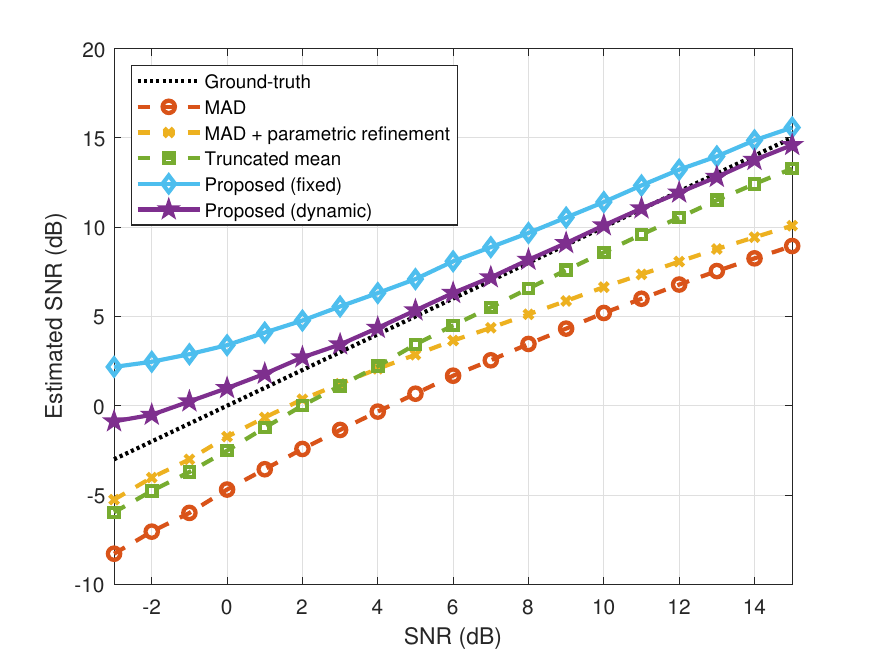}
    \caption{Estimated SNR depending on SNR.}
    \label{fig:snr}
\end{figure}

Table~\ref{tbl:simtime} presents the MATLAB operation time and the normalized time comparison. The proposed method incurs only a 10-20\% increase in computational time compared to MAD-based estimators, while achieving a shorter simulation time than that of the estimator in~\cite{ref:lowreschannelest}. This indicates that the proposed approach attains higher accuracy without a significant increase in software simulation time, while maintaining similar operation time.

\begin{table*}[]
    \centering
    \caption{MATLAB operation time and normalized time comparison}
    \label{tbl:simtime}
    \begin{tabular}{|c|c|c|c|c|}
    \hline
    Algorithm                 & MAD   & MAD + parametric refinement & Truncated mean & Proposed \\ \hline\hline
    Operation time ($\mu$s)   & 70.75 & 77.89                       & 95.07          & 88.11    \\ 
    Normalized operation time & 1.0   & 1.1                         & 1.3            & 1.2      \\ \hline
    \end{tabular}
\end{table*}

\section{Hardware Design and Implementation}\label{sec:hardware}
In this section, we propose a VLSI architecture for blind parameter estimation, which serves as a hardware implementation of the low-complexity blind estimators for the average noise power, the average signal power, and SNR, devised in Sec.~\ref{sec:proposed}. We also provide the corresponding FPGA implementation results. 

\subsection{VLSI Architecture}
Fig.~\ref{fig:highlevelvlsi} illustrates a high-level depiction of the proposed VLSI architecture, which realizes the hardware implementation of the algorithm described in the previous section.
The architecture is composed of five primary modules: \textit{(i)} an FFT unit for beamspace transformation, \textit{(ii)} an element-wise magnitude squaring operator, \textit{(iii)} a noise power estimator unit, \textit{(iv)} a signal power estimator unit, and \textit{(v)} an SNR estimator unit.
The FFT unit maps the received antenna domain signal vector $\mathbf{y}$ to its beamspace representation $\bar{\mathbf{y}}$.
The element-wise squaring operator transforms the complex beamspace signal vector $\bar{\mathbf{y}}$ into the squared magnitude vector $|\bar{\mathbf{y}}|^2$.  
The noise power estimator unit performs the average noise power estimation following the procedure in Algorithm~\ref{alg:noisepwrest}, employing sorting and threshold-based separation with corresponding units, respectively. The signal power estimator unit derives the signal power estimate using the outputs of the noise power estimator unit, namely the cumulative sum $S_M$ (equivalent to $\|\bar{\mathbf{y}}\|^2$) and the estimated noise power $\widehat{N}_0$, as shown in \eqref{eqn:signalpowerest}. The SNR estimator unit then evaluates the SNR based on these estimated parameters as shown in \eqref{eqn:snrest}. 
To achieve high throughput, the architecture adopts a streaming structure. Specifically, it processes a new signal vector and produces the corresponding estimates for the new incoming signal samples at every clock cycle. 

To improve hardware efficiency, a two’s complement fixed-point representation is adopted. The word lengths are determined through performance evaluation to achieve near floating-point accuracy while minimizing resource usage.
The antenna-domain signal elements are represented using 16 bits with 8 fractional bits, whereas the beamspace signal elements use 10 bits with 8 fractional bits due to scaling in the FFT process. The output of the element-wise squaring unit is represented with 16 bits and 8 fractional bits.
For the noise power estimate $\widehat{N}_0$, 16 bits with 8 fractional bits are used. The same format, consisting of 16 bits with 8 fractional bits, is applied to both the total noisy signal power $S_M$ and the signal power estimate $\widehat{P}_x$. Finally, the SNR estimate $\hat{\rho}$ is represented using 24 bits with 8 fractional bits.

\begin{figure*}
    \centering
    \includegraphics[width=0.8\linewidth]{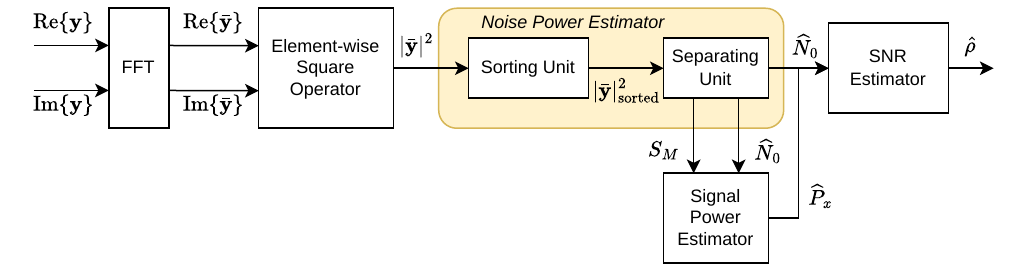}
    \caption{High-level VLSI architecture of the proposed algorithm.}
    \label{fig:highlevelvlsi}
\end{figure*}

Fig.~\ref{fig:preprocess} illustrates the architecture of the FFT and the element-wise squaring operator employed for preprocessing. For the FFT implementation, a Xilinx FFT IP core with radix-2 pipelined I/O streaming was adopted, enabling the input and output of one complex vector element per clock cycle. 
Consequently, the FFT core outputs $M$ complex beamspace elements over $M$ clock cycles in a sequential manner. These complex-valued samples are then processed by an element-wise squaring operator, which operates sequentially on the incoming data. Since the operation computes the sum of the squares of the real and imaginary components, the resulting values are always non-negative and are thus represented in an unsigned format. Also, the resulting values are generated sequentially and fed directly into the sorting unit.

\begin{figure}
    \centering
    \includegraphics[width=\linewidth]{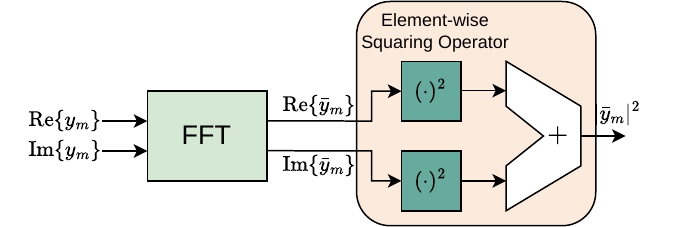}
    \caption{Architecture of the unit for preprocessing; FFT and element-wise square operator.}
    \label{fig:preprocess}
\end{figure}

Fig.~\ref{fig:noisepwrunit} depicts the architecture of the noise power estimator unit, which consists of a sorting unit and a separating unit. 
The proposed sorting unit adopts a systolic insertion-based architecture, where sorting is progressively achieved as data propagate through a cascade of registers. Its operation is controlled by a finite-state machine (FSM), which manages three phases: data loading, pipeline flushing, and sorted output generation. During the loading phase, input samples are sequentially provided at one sample per clock cycle, while each stage performs a local compare-and-swap operation with the incoming value. Specifically, the smaller value is retained in the current stage, while the larger value is forwarded to the next stage, allowing partially sorted results to form across the pipeline. After all inputs are received, the FSM initiates a flushing phase to complete data propagation, followed by an output phase where the sorted values are sequentially read out. This structure avoids long combinational paths by restricting each stage to a single comparison per cycle, enabling timing-efficient hardware implementation.

Subsequently, as illustrated in Fig.~\ref{fig:noisepwrunit}, the sorted input samples are sequentially streamed into the separating unit. 
As each sample arrives, it is used to compute the finite difference $\Delta [|\bar{\mathbf{y}}|^2_\mathrm{sorted}]_m$ with the previously stored sample. The current sample is then registered for the next comparison, while the temporary cumulative sum $S_m$ is updated by adding the previous sample. Based on these values, the comparison in \eqref{eqn:thresholdcheck} is performed. 
The threshold parameter $\gamma$ is constrained to be a power-of-two, allowing the multiplication to be efficiently implemented using bit-shifting operations. 
This procedure corresponds to lines 4 through 10 of Algorithm~\ref{alg:noisepwrest}.
When a \textit{hit} is detected, the corresponding $S_{m^*}$ is recorded and subsequently normalized to obtain the average. Instead of performing a direct division by $m^*$, which would impose significant hardware complexity, a lookup table (LUT) is employed to retrieve the precomputed value of $1/m^*$, and the normalization is carried out via multiplication. 
Finally, the architecture outputs both the total accumulated sum $S_M$, required for signal power estimation, and the estimated noise power $\widehat{N}_0$.

\begin{figure*}
    \centering
    \includegraphics[width=0.8\linewidth]{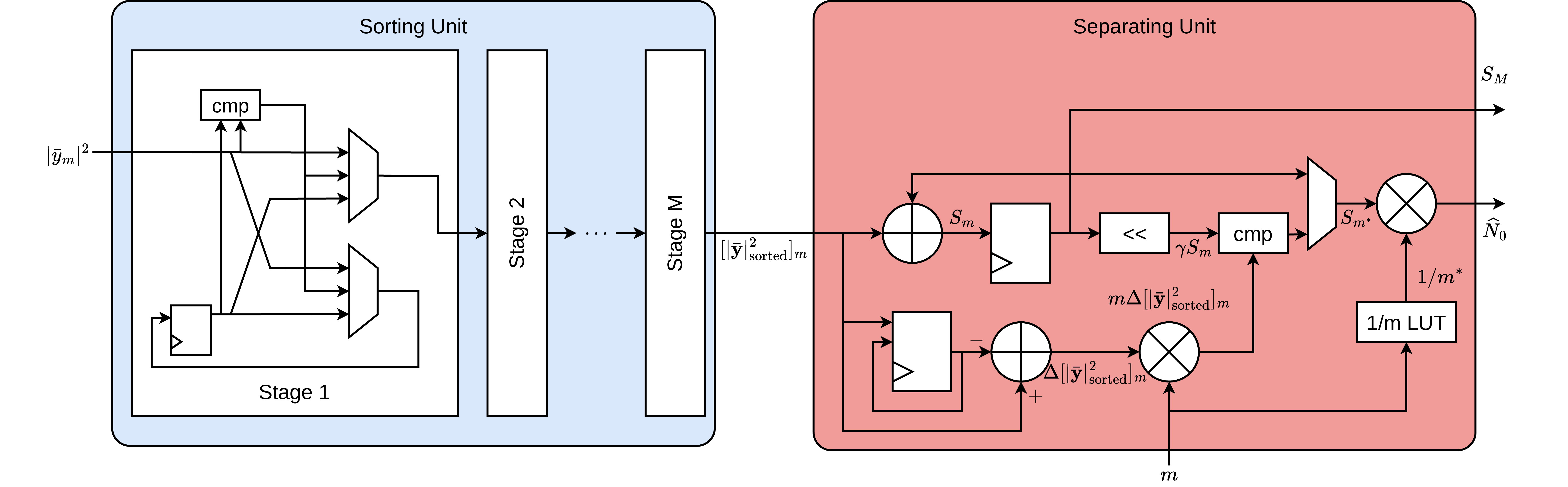}
    \caption{Architecture of the noise power estimation unit.}
    \label{fig:noisepwrunit}
\end{figure*}

Fig.~\ref{fig:sigpow_snr_vlsi} illustrates the architecture of the signal power estimator unit and the SNR estimator unit. The signal power estimator receives $S_M$, which is equivalent to $\|\bar{\mathbf{y}}\|^2$, and $\widehat{N}_0$ from the separating unit and computes the signal power estimate as defined in \eqref{eqn:signalpowerest}. In this process, the division by $M$ is implemented using a bit-shifting operation, under the assumption that the number of antennas is a power of two.
Furthermore, the SNR estimator unit performs the SNR computation defined in \eqref{eqn:snrest}, which is implemented as a simple division operation.

\begin{figure}
    \centering
    \includegraphics[width=0.8\linewidth]{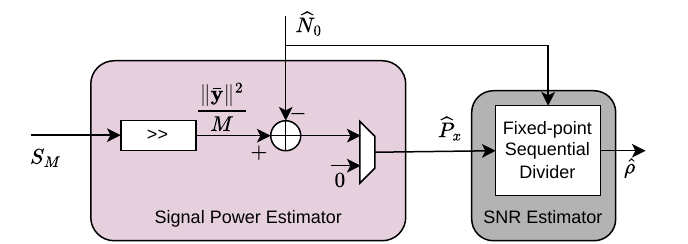}
    \caption{Architecture of the signal power estimation unit and SNR calculation unit.}
    \label{fig:sigpow_snr_vlsi}
\end{figure}

\subsection{FPGA Implementation Results}
To validate the practical effectiveness of the proposed method, the VLSI design is implemented on the AMD-Xilinx Kintex UltraScale+ FPGA KCU116 Evaluation Kit. The implementation is carried out for various numbers of BS antennas, and the results are summarized in Table~\ref{tbl:resource1}.
Although the resource utilization, including LUTs and flipflops, increases with $M$, the growth is sublinear. Also, even for $M=128$, the proposed architecture occupies only a small fraction of the available hardware resources. This indicates that the proposed algorithm achieves low complexity with low computational overhead. The required clock cycles also increase with $M$ because the beamspace vector elements are processed sequentially across the functional units. Nevertheless, the estimated latency is approximately 2.182 $\mu$s even for the largest configuration ($M=128$). Considering that the minimum 5G transmission time interval (TTI) is 125 $\mu$s and the duration of the single OFDM symbol in that configuration is approximately 8.9 $\mu$s~\cite{ref:3gpp38211,ref:3gpp38214}, the proposed architecture can estimate the parameters before the next symbol is received. Overall, the implementation results indicate that the proposed method not only enables blind estimation of multiple system parameters from a single signal sample but also demonstrates efficiency in terms of complexity, hardware resource usage, and latency.

Table~\ref{tbl:resource2} presents the breakdown of resource utilization across the processing units. Here, ``P'' denotes the preprocessing units, including the FFT and the element-wise squaring unit; ``N'' represents the noise power estimation units, consisting of the sorting and separating units; and ``S'' corresponds to the signal power estimator and the SNR estimator units. The overall latency is primarily dominated by the preprocessing and noise power estimation stages. In the preprocessing stage, the FFT accounts for a significant portion of the hardware resources. In the noise power estimation stage, both the sorting and separating units contribute substantially to the latency. Moreover, as the dimension $M$ increases, the latency of these units also grows, while the hardware resource utilization exhibits a sublinear scaling trend.
On the other hand, the signal power estimator and the SNR estimator units exhibit nearly identical operational structures across different $M$, resulting in similar hardware resource utilization and identical clock cycle requirements across all configurations.

\begin{table*}[] 
    \caption{FPGA hardware resource utilization results for different $M$ implemented on the Xilinx Kintex UltraScale+ FPGA KCU116 Evaluation Kit}\label{tbl:resource1}
    {\centering
    \begin{tabular}{|c|c|c|c|c|}
    \hline
    The number of BS antennas $M$ & 16             & 32             & 64 & 128  \\ \hline\hline
    LUTs                          & 1499 (0.69\%)  & 2291 (1.06\%) & 3851 (1.77\%)  & 5711 (2.63\%)    \\ 
    LUTs as logic                 & 1391 (0.64\%)  & 2142 (0.99\%) & 3662 (1.69\%) &  5448 (2.51\%)    \\ 
    LUTs as memory                &  108 (0.11\%)  & 149 (0.15\%)  & 189 (0.18\%) &  263 (0.26\%)    \\ 
    Flipflops                     & 2032  (0.47\%) & 2940 (0.68\%) & 4288 (0.99\%) &  6860 (1.58\%)  \\ 
    DSP48 units                   &    4  (0.22\%) &  7 (0.38\%)   &  7 (0.38\%)  & 10 (0.54\%)    \\ 
    Maximum clock frequency (MHz) & 440 &  440  & 359 & 357  \\ 
    Latency (clock cycles)        & 180 &  280  & 453 & 779 \\ 
    Latency ($\mu$s)              & 0.409  & 0.636 & 1.262 & 2.182 \\ 
    Power consumption$^a$ (W)         & 0.036  & 0.047 & 0.065 & 0.097 \\ 
    Maximum throughput$^b$ (Mvectors/s)       & 27.500  & 13.750 & 5.609 & 2.789 \\
    Maximum throughput/LUTs               & 18345   & 6002 & 1090  & 244 \\
    \hline
    \end{tabular}\par}
    \vspace{2mm}
    {\footnotesize
    $^{a}$The power consumption is obtained from the statistical dynamic power estimation at the maximum clock frequency under a supply voltage of 0.85 V. \\
    $^{b}$The maximum throughput is obtained in million vectors processed per second and is calculated as $f_\mathrm{max}/M$, where $f_\mathrm{max}$ denotes the maximum clock frequency.}
\end{table*}

\begin{table*}[]
    \centering
    \caption{FPGA hardware resource utilization breakdown for different $M$ implemented on the Xilinx Kintex UltraScale+ FPGA KCU116 Evaluation Kit }\label{tbl:resource2}
    \begin{tabular}{|c|ccc|ccc|ccc|ccc|}
    \hline
    $M$                    & \multicolumn{3}{c|}{16}               & \multicolumn{3}{c|}{32}               & \multicolumn{3}{c|}{64}               & \multicolumn{3}{c|}{128}              \\ 
    Modules                & P & N & S & P & N & S & P & N & S & P & N & S \\ \hline\hline
    LUTs                   & 692   & 703       & 135              & 862   & 1315       & 138              & 1025   & 2722       & 135              & 1244  & 4347       & 143              \\ 
    LUTs as logic          & 584   & 703       & 135              & 713   & 1315       & 138              & 836   & 2722       & 135              & 981   & 4347       & 143              \\ 
    LUTs as memory         & 108   & 0          & 0                & 149   & 0          & 0                & 189   & 0          & 0                & 263   & 0          & 0                \\ 
    Flipflops              & 1037  & 742        & 253              & 1382  & 1296       & 262              & 2248  & 2399       & 277              & 1985  & 4582       & 293              \\ 
    DSP48 units            & 3     & 1          & 0                & 6     & 1          & 0                & 6     & 1          & 0                & 9     & 1          & 0                \\ 
    Latency (clock cycles) & 84    & 57         & 39               & 143   & 98        & 39               & 213   & 201        & 39               & 351   & 389        & 39               \\ \hline
    \end{tabular}
\end{table*}

To the best of our knowledge, there are no existing hardware implementations of blind estimators for noise power, signal power, and SNR in multi-antenna communication systems that allow for a direct and fair comparison. 
While such directly comparable hardware implementations of single-sample-based estimators for multi-antenna systems do not exist, several prior studies have investigated hardware realizations of blind noise power and SNR estimation algorithms for other systems.
For example, the work in \cite{ref:blindsnrhw} proposed a blind SNR estimator capable of operating with a single sample by exploiting modulation-specific properties for single-antenna systems, and demonstrated its implementation on both FPGA and application-specific integrated circuit (ASIC) platforms. However, detailed results for FPGA synthesis and the chip measurements were not reported. In addition, the estimator exhibits limited performance when operating with a single sample and requires multiple samples to achieve stable performance. Furthermore, it is not designed for multi-antenna receivers. In contrast, the proposed method achieves high estimation accuracy even with a single sample, while maintaining low algorithmic complexity, hardware resource utilization, and latency.

\section{Conclusion}\label{sec:conclusion}
In this paper, we proposed a low-complexity blind estimator for the average noise power, average signal power, and SNR tailored to mmWave massive multi-antenna uplink systems. The key contribution lies in enabling reliable parameter estimation from only a single received signal sample, eliminating the need for pilot signaling, iterative processing, training phases for machine learning, or multiple observations. 
By exploiting beamspace sparsity inherent to mmWave channels, the proposed method performs implicit separation of signal and noise components through a sorting-based mechanism with a finite-difference criterion. In particular, the method leverages the order statistics of noise power under Gaussian assumptions to justify the separation process and enable statistically grounded estimation. This leads to accurate parameter estimation with minimal computational overhead. In addition to its algorithmic simplicity, the proposed approach is explicitly designed for hardware efficiency, resulting in a streamlined architecture that avoids complex operations and is well-suited for real-time implementation. To validate its practical viability, we developed a VLSI architecture and demonstrated its effectiveness through FPGA implementation on an AMD-Xilinx Kintex UltraScale+ platform. The implementation results confirm that the proposed estimator achieves low-latency operation within a time shorter than a single OFDM symbol while exhibiting favorable scalability in hardware resource usage. 
Overall, this work establishes that accurate and real-time blind parameter estimation is achievable with minimal complexity, while consistently achieving higher estimation accuracy than existing single-sample-based methods, making the proposed approach a strong candidate for mmWave multi-antenna systems.

\appendices

\section{Proof of Lemma~\ref{lem:monotone}}\label{app:prooflemma}
Since $\left[|\bar{\mathbf{y}}|^2_{\mathrm{sorted}}\right]_m>0$ for all $m$ and is monotonically non-decreasing, we have
\begin{equation}
\begin{aligned}
    \hat{\mu}_m & = \frac{1}{m}\left((m-1)\hat{\mu}_{m-1} + \left[|\bar{\mathbf{y}}|^2_{\mathrm{sorted}}\right]_m\right) \\
    & = \frac{m-1}{m}\hat{\mu}_{m-1} + \frac{\left[|\bar{\mathbf{y}}|^2_{\mathrm{sorted}}\right]_m}{m} \\
    & = \hat{\mu}_{m-1} - \frac{1}{m}\hat{\mu}_{m-1}+ \frac{\left[|\bar{\mathbf{y}}|^2_{\mathrm{sorted}}\right]_m}{m} \\
    & \geq \hat{\mu}_{m-1},
\end{aligned}
\end{equation}
where the last inequality holds because
\begin{equation}
    -\frac{1}{m}\hat{\mu}_{m-1} + \frac{\left[|\bar{\mathbf{y}}|^2_{\mathrm{sorted}}\right]_m}{m} \geq 0,
\end{equation}
which follows from
\begin{equation}
\begin{aligned}
    \hat{\mu}_{m-1} & = \frac{1}{m-1}\sum_{m'=1}^{m-1} \left[|\bar{\mathbf{y}}|^2_{\mathrm{sorted}}\right]_{m'} \\
    & = \frac{1}{m-1}\big( \left[|\bar{\mathbf{y}}|^2_{\mathrm{sorted}}\right]_{1} + \cdots + \left[|\bar{\mathbf{y}}|^2_{\mathrm{sorted}}\right]_{m-1} \big) \\
    & \leq \frac{1}{m-1}
\left(
\underbrace{
\left[|\bar{\mathbf{y}}|^2_{\mathrm{sorted}}\right]_m + \cdots + \left[|\bar{\mathbf{y}}|^2_{\mathrm{sorted}}\right]_m
}_{(m-1) \text{ terms}}
\right) \\
    & = \left[|\bar{\mathbf{y}}|^2_{\mathrm{sorted}}\right]_m.
\end{aligned}
\end{equation}
Hence, $\hat{\mu}_m$ is monotonically non-decreasing in $m$. Therefore, since $\hat{N}_0=\hat{\mu}_{m^*}$ and $\hat{N}_0^{\mathrm{ora}}=\hat{\mu}_{m_0}$, the stated inequalities follow directly. \qed

\section{Proof of Proposition~\ref{prop:unbiased}}\label{app:proofprop}
Under perfect separation, the first $(M-k)$ elements of $|\bar{\mathbf y}|^2_{\mathrm{sorted}}$ correspond exactly to the noise-only components. Hence,
\begin{equation}
    \hat N_0^{\mathrm{ora}}
    =
    \frac{1}{M-k}\sum_{i\in\mathcal I_n} |\bar y_i|^2.
\end{equation}
For every $i\in\mathcal I_n$, we have $\bar y_i=\bar n_i$. Since the beamspace noise remains AWGN after the unitary DFT transformation, each $\bar n_i$ satisfies
\begin{equation}
    \bar n_i \sim \mathcal{CN}(0,N_0),
\end{equation}
which implies
\begin{equation}
    \mathbb E[|\bar n_i|^2]=N_0.
\end{equation}
Therefore,
\begin{equation}
    \mathbb E[\hat N_0^{\mathrm{ora}}]
    =
    \frac{1}{M-k}\sum_{i\in\mathcal I_n}\mathbb E[|\bar n_i|^2]
    =
    \frac{1}{M-k}(M-k)N_0
    =
    N_0.
\end{equation}
This completes the proof. \qed

\section{Proof of Corollary~\ref{cor:sparsity}}\label{app:proofcor}
From the proof of Proposition \ref{prop:unbiased},
\begin{equation}
    \hat N_0^{\mathrm{ora}}
    =
    \frac{1}{M-k}\sum_{i\in\mathcal I_n} |\bar n_i|^2,
\end{equation}
where $|\bar n_i|^2$ are i.i.d. exponential random variables with variance $N_0^2$. Thus,
\begin{equation}
\begin{aligned}
    \mathrm{Var}(\hat N_0^{\mathrm{ora}})
    & =
    \frac{1}{(M-k)^2}\sum_{i\in\mathcal I_n}\mathrm{Var}(|\bar n_i|^2)
    \\& =
    \frac{1}{(M-k)^2}(M-k)N_0^2
    \\& =
    \frac{N_0^2}{M-k}.
\end{aligned}
\end{equation}
Since $(M-k)$ increases as $k$ decreases, the variance decreases as the sparsity increases. 

Now, as $M \to \infty$, if the fraction of signal-bearing components does not converge to one, i.e., if
\begin{equation}
    \frac{k}{M} \to \alpha
    \qquad \text{for some } \alpha \in [0,1),
\end{equation}
then
\begin{equation}
    M-k
    =
    M\left(1-\frac{k}{M}\right)
    \sim (1-\alpha)M \to \infty.
\end{equation}
Hence,
\begin{equation}
    \mathrm{Var}(\hat N_0^{\mathrm{ora}})
    =
    \frac{N_0^2}{M-k}
    \to 0.
\end{equation}

Since Proposition 1 shows that the estimator is unbiased, we have
\begin{equation}
    \mathbb E\!\left[(\hat N_0^{\mathrm{ora}}-N_0)^2\right]
    =
    \mathrm{Var}(\hat N_0^{\mathrm{ora}})
    \to 0.
\end{equation}
Therefore, $\hat N_0^{\mathrm{ora}}$ converges in mean square, and hence in probability, to $N_0$, which establishes asymptotic consistency. \qed

\ifCLASSOPTIONcaptionsoff
  \newpage
\fi



\bibliographystyle{IEEEtran}
\bibliography{bibtex/bib/IEEEabrv,bibtex/bib/references}
\end{document}